\documentclass[letterpaper]{article} 
\usepackage[draft]{aaai25}  
\usepackage{times}  
\usepackage{helvet}  
\usepackage{courier}  
\usepackage[hyphens]{url}  
\usepackage{graphicx} 
\usepackage{natbib}  
\usepackage{caption} 
\usepackage{algorithm}
\usepackage{algorithmic}

\usepackage{newfloat}
\usepackage{listings}
\DeclareCaptionStyle{ruled}{labelfont=normalfont,labelsep=colon,strut=off} 
\floatstyle{ruled}
\newfloat{listing}{tb}{lst}{}
\floatname{listing}{Listing}
\usepackage[usenames,svgnames,table]{xcolor}
\definecolor{myblue}{rgb}{0.1,0.1,0.5}
\definecolor{mygreen}{rgb}{0.00,0.26,0.15}
\definecolor{goldenpoppy}{rgb}{0.99, 0.76, 0.0}
\definecolor{iris}{rgb}{0.35, 0.31, 0.81}

\usepackage{xspace}
\usepackage{framed}
\usepackage{xifthen}
\usepackage{tikz}
\usetikzlibrary{calc}
\usepackage{amsmath}
\usepackage{amsthm}
\usepackage{amssymb}
\usepackage{nicefrac}
\usepackage{algorithmic}
\usepackage{algorithm}
\usepackage{multirow}

\usepackage{lineno}
\usepackage{paralist}
\usepackage{todonotes}

\tikzstyle{smallvertex}=[circle, draw=black, fill=white, inner sep=2pt]

\makeatletter
\newcommand{\letter}[1]{\@alph{#1}}
\makeatother

\definecolor{shadecolor}{gray}{0.75}

\usepackage{comment}

\usepackage{bm} 

\usepackage[appendix=strip]{apxproof}

\newtheoremrep{lemma}{Lemma}[section]

\newtheoremrep{theorem}{Theorem}[section]

\newtheorem{observation}[lemma]{Observation}

\newtheorem{claim}[lemma]{Claim}

\theoremstyle{definition}
\newtheorem{example}[lemma]{Example}
\newtheorem{remark}[lemma]{Remark}
\newtheorem{definition}[lemma]{Definition}

\newlength{\RoundedBoxWidth}
\newsavebox{\GrayRoundedBox}
\newenvironment{GrayBox}[1]%
   {\setlength{\RoundedBoxWidth}{.93\textwidth}
    \def\boxheading{#1}
    \begin{lrbox}{\GrayRoundedBox}
       \begin{minipage}{\RoundedBoxWidth}}%
   {   \end{minipage}
    \end{lrbox}
    \begin{center}
    \begin{tikzpicture}%
       \node(Text)[draw=black!20,fill=white,rounded corners,%
             inner sep=2ex,text width=\RoundedBoxWidth]%
             {\usebox{\GrayRoundedBox}};
        \coordinate(x) at (current bounding box.north west);
        \node [draw=white,rectangle,inner sep=3pt,anchor=north west,fill=white] 
        at ($(x)+(6pt,.75em)$) {\boxheading};
    \end{tikzpicture}
    \end{center}}     
\newenvironment{defproblemx}[2][]{\noindent\ignorespaces%
                                \FrameSep=6pt%
                                \parindent=0pt%
                \vspace*{-1.5em}
                \ifthenelse{\isempty{#1}}{%
                  \begin{GrayBox}{\textsc{#2}}%
                }{%
                  \begin{GrayBox}{\textsc{#2}  parameterized by~{#1}}%
                }
                \begin{tabular*}{\textwidth}{@{\hspace{.1em}} >{\itshape} p{1.8cm} p{0.8\textwidth} @{}}%
            }{
                \end{tabular*}%
                \end{GrayBox}%
                \ignorespacesafterend
            }

\newcommand{\FPT}{\textsf{FPT}}

\newcommand{\NPH}{\textsf{NP}-hard}

\newcommand{\card}[1]{\ensuremath{{\vert {#1} \vert }}}
\newcommand{\ca}[1]{\ensuremath{\mathcal{#1} }}
\newcommand{\set}[1]{\ensuremath{\left\{ {#1} \right\}}}
\newcommand{\fn}[3]{\ensuremath{{{#1} : {#2} \rightarrow {#3}}}}

\newcommand{\cO}{\mathcal{O}}

\newcommand{\poly}{\ensuremath{\text{poly}}}

\mathchardef\mh="2D

\newcommand{\cc}{$c$-closed}

\newcommand{\cl}{\ensuremath{\text{cl}}}

\newcommand{\lt}{\ensuremath{\Lambda}} %
\newcommand{\pngt}{\ensuremath{\Delta_0}} 
\newcommand{\pone}{\ensuremath{\Delta_1}} 
\newcommand{\ptwo}{\ensuremath{\Delta_2}} 
\newcommand{\pcq}{\ensuremath{\Delta}} 

\newcommand{\ppx}{\ensuremath{(\pcq, k + 1)}} 

\newcommand{\pdc}{\ensuremath{(\pcq, k)}} 

\newcommand{\tc}{\ensuremath{(\pngt, \pone, \ptwo, c)}} %

\newcommand{\tg}{\ensuremath{(\pngt, \pone, \ptwo, \gamma)}} 

\newcommand{\tp}{\ensuremath{(\pngt, \pone, \ptwo)}} 

\newcommand{\nbdbound}{\ensuremath{c - 1 + 2\eta(\pcq - \pone)}}
\newcommand{\nbdboundplus}{\ensuremath{c - 1 + 2\eta(\pcq + 1 - \pone)}}

\newcommand{\weaknbdboundplus}{\ensuremath{\gamma - 1 + 2\eta(\pcq + 1 - \pone)}}

\newcommand{\mcbound}{\ensuremath{3 \cdot 2^{\nbdboundplus} \cdot n^2 \cdot \lt}} 

\newcommand{\weakmcbound}{\ensuremath{3 \cdot 2^{\weaknbdboundplus} \cdot n^2 \cdot \lt}} 

\newcommand{\mplexbound}{\ensuremath{4 \cdot 2^{\nbdboundplus} \cdot n^{\max\set{2k, k + 2}} \cdot \lt}}

\newcommand{\weakmplexbound}{\ensuremath{4 \cdot 2^{\weaknbdboundplus} \cdot n^{\max\set{2k, k + 2}} \cdot \lt}}

\newcommand{\mdcbound}{\ensuremath{4 \cdot 2^{\nbdboundplus} \cdot n^{k + 2} \cdot \lt}}

\newcommand{\weakmdcbound}{\ensuremath{4 \cdot 2^{\weaknbdboundplus} \cdot n^{k + 2} \cdot \lt}}

\title{Temporal Triadic Closure: Finding Dense Structures in \\ Social Networks That Evolve}
\author {
    Tom Davot, 
    Jessica Enright, 
    Jayakrishnan Madathil, 
    Kitty Meeks 
}
\affiliations {
    University of Glasgow, Glasgow, United Kingdom\\
    tom.davot@glasgow.ac.uk, jessica.enright@glasgow.ac.uk, jayakrishnan.madathil@glasgow.ac.uk,
    kitty.meeks@glasgow.ac.uk
}

\begin{document}

\maketitle

\begin{abstract}
A graph $G$ is $c$-closed if every two vertices with at least $c$ common neighbors are adjacent to each other. Introduced by Fox, Roughgarden, Seshadhri, Wei and Wein [ICALP 2018, SICOMP 2020], this definition is an abstraction of the triadic closure property exhibited by many real-world social networks, namely, friends of friends tend to be friends themselves. Social networks, however, are often temporal rather than static---the connections  change over a period of time. And hence temporal graphs, rather than static graphs, are often better suited to model social networks. Motivated by this, we introduce a definition of temporal $c$-closed graphs, in which if two vertices $u$ and $v$ have at least $c$ common neighbors during a short interval of time, then $u$ and $v$ are adjacent to each other around that time. Our pilot experiments show that several real-world temporal networks are $c$-closed for rather small values of $c$. We also study the computational problems of enumerating maximal cliques and similar dense subgraphs in temporal $c$-closed  graphs; a clique in a temporal graph is a subgraph that lasts for a certain period of time, during which every possible edge in the subgraph becomes active often enough, and other dense subgraphs are defined similarly. We bound the number of such maximal dense subgraphs  in a temporal $c$-closed graph that evolves slowly, and thus show that the corresponding enumeration problems admit efficient algorithms; by slow evolution, we mean that between consecutive time-steps, the local change in adjacencies remains small. Our work also adds to a growing body of literature on defining suitable structural parameters for temporal graphs that can be leveraged to design efficient algorithms. 
\end{abstract}

%

\section{Introduction} 
   
Social networks evolve.  Influencers gain and lose followers on social media; 
ants in a colony guide each other to food; scientists collaborate with their peers. All these examples involve networks in which  connections are created and destroyed as time passes. Moreover, even when a relationship within a network is  continuous, the interactions that provide evidence of that relationship---such as being in the same location as a friend, or exchanging an email---only happen at some points in time.  All these considerations mean that temporal information must be taken into account to gain a full understanding of many social networks.  Temporal graphs provide a useful formalism for modeling these evolving networks; in this work we adopt the widely-used model of~\citeauthor{DBLP:journals/jcss/KempeKK02}~\shortcite{DBLP:journals/jcss/KempeKK02}, in which a temporal graph $\ca{G}$ consists of a static graph $G$ called the \emph{underlying graph} or the \emph{footprint} of $\ca{G}$, together with a function (which we will denote by $\lambda$ in this paper) assigning to each edge a (finite) subset of $\mathbb{N}$ representing the discrete timesteps at which the edge  appears in the graph or is \emph{active}.  We can equivalently consider a temporal graph to be a sequence of static graphs or \emph{snapshots}, where the $t^{th}$ snapshot contains all vertices and only those edges that are active at time $t$. 

Modeling social networks in this way, however, introduces a new level of algorithmic challenge.  Even problems that are polynomial-time solvable on static graphs often become \NPH\ when a temporal dimension is added; in some cases this holds even when very strong restrictions are placed on the footprint of the graph, such as requiring it to be a path or a star \cite{DBLP:journals/jcss/MertziosMNZZ23,DBLP:journals/jcss/AkridaMSR21}.  This has prompted significant recent interest in the design of temporal graph parameters that additionally take into account temporal information, with the hope of identifying situations in which natural problems admit efficient algorithms. None of these, however, seems particularly well suited for the design of efficient algorithms to solve problems on social networks.  Some only offer a limited benefit over restricting the structure of the footprint: the \emph{timed feedback vertex number} \cite{DBLP:journals/algorithmica/CasteigtsHMZ21}, \emph{temporal feedback edge/connection number} \cite{DBLP:journals/dam/HaagMNR22}, and the various temporal analogues of treewidth \cite{DBLP:conf/birthday/FluschnikMNRZ20} are guaranteed to take small values when the footprint is a tree, so they cannot hope to give tractability for those problems that remain \NPH\ under this (or a stronger) restriction.  Restricting others, such as the \emph{vertex-interval-membership-width} \cite{DBLP:journals/algorithmica/BumpusM23} has proved more widely useful in making problems tractable, but this success comes at a price: the vertex-interval-membership-width of a temporal graph representing a social network will only be small if, at every timestep, most individuals have either (i) never yet formed any connections, or (ii) will never again form a connection, an assumption that seems unrealistic for the vast majority of social networks. 

In general, most of the temporal graph parameters defined so far require that the graph be sparse (in the sense of having only a linear number of edges active) at every timestep. 
Very recent work \cite{DBLP:journals/corr/abs-2404-19453} has tried to address this limitation by introducing temporal analogues of the parameter cliquewidth, modular width and neighborhood diversity, which can take small values on dense temporal graphs that are sufficiently highly structured, but these are not without their downsides.  Computing the temporal cliquewidth of a temporal graph is an \NPH\ problem in itself \cite{DBLP:journals/corr/abs-2404-19453}, making it hard to gain intuition about which real-world networks, if any, might have small values of this parameter.  Meanwhile, for the temporal modular width or temporal neighborhood diversity of a social network to be small there must be large groups of vertices between which the interactions are uniform (i.e. every member of each group has identical interactions with each member of the other group at every time-step), which again does not seem a credible property for many social networks.

In this paper we introduce a new structural parameter for temporal graphs, namely the notion of temporal \cc\ graphs. This is an adaptation of static \cc\ graphs introduced by \citeauthor{DBLP:journals/siamcomp/FoxRSWW20}~\shortcite{DBLP:conf/icalp/FoxRSWW18,DBLP:journals/siamcomp/FoxRSWW20}; we discuss the usefulness of these graphs as a model for social networks, and present some preliminary results on real-world networks in in Section~\ref{sec:definitions}.  We also show that in order to replicate many of the algorithmic results about static \cc\ graphs we need to impose further restrictions, and in Section~\ref{sec:stability} we introduce an additional parameter that captures the extent to which the network changes locally between one timestep and the next.  In Sections~\ref{sec:bounds} and ~\ref{sec:bounds-more}, we present our main theoretical results---bounds for the number of maximal cliques and similar dense subgraphs in a temporal \cc\ graph that evolves sufficiently slowly, which imply the existence of efficient algorithms to find all such subgraphs. %

\paragraph*{Terminology and notation.} For a (static) graph $G$, we use $V(G)$ and $E(G)$ to denote the vertex set and edge set of $G$, respectively. A temporal graph $\ca{G}$ is a pair $(G, \lambda)$, where $G$ is a static graph and the function $\fn{\lambda}{E(G)}{2^{\mathbb{N}}}$ specifies the discrete time-steps at which each edge $e$ of $G$ is active. 
We assume throughout that $\lambda(e)$ is finite. We also assume that the lifespan of a temporal graph $\ca{G}$ starts at time-step $1$. We use $\lt_{\ca{G}}$ (or simply $\lt$ when $\ca{G}$ is clear) to denote the maximum time-step at which any edge is active, and call $\lt_{\ca{G}}$ the lifetime of $\ca{G}$. 
We call $G$ the footprint or the underlying graph of $\ca{G}$. 
For $\ca{G} = (G, \lambda)$ and a vertex $v \in V(G)$, we use $\ca{G} - v$ to denote the temporal graph obtained from $\ca{G}$ by deleting $v$, i.e., $\ca{G} - v = (G - v, \lambda')$, where $G - v$ is the subgraph of $G$ induced by $V(G) \setminus \set{v}$ and $\lambda'$ is the restriction of $\lambda$ to $E(G - v)$. 
      By a time-interval $I \subseteq \mathbb{N}$ we mean a set of consecutive time-steps, i.e., $I = [a, b] = \set{a, a + 1, a + 2,\ldots, b}$ for some $a, b \in \mathbb{N}$, where $a \leq b$. The length of the time-interval $[a, b]$ is $b - a + 1$, i.e., the number of time-steps in $[a, b]$. For a time-interval $I$ and a temporal graph $\ca{G} = (G, \lambda)$, we use $G_I$ to denote the subgraph of $G$ that consists of all the edges of $G$ that are active at some time-step in $I$, i.e., $V(G_I) = V(G)$ and $E(G_I) = \set{e \in E(G) ~|~ \lambda(e) \cap I \neq \emptyset}$. 
    For $u, v \in V(G)$, we say that $u$ and $v$ are adjacent to each other \emph{during} $I$ if $uv \in E(G_I)$;  just to emphasize,  $u$ and $v$ are adjacent during $I$ if $u$ and $v$ are adjacent to each other at \emph{some} time-step in $I$, and not necessarily at \emph{every} time-step in $I$. 
        For $v \in V(G)$, we use $N_I(v)$ to denote the set of neighbors of $v$ in the graph $G_I$; when $I = \set{t}$, we omit the braces and simply write $N_t(v)$. For $u, v \in V(G)$, we use $CN_I(u, v)$ to denote the common neighborhood of $u$ and $v$ in the graph $G_I$, i.e., $CN_I(u, v) = N_I(u) \cap N_I(v)$. Notice that this definition  does not require that a vertex $w \in CN_{I}(u, v)$ be adjacent to both $u$ and $v$ at the same time-step; we have $w \in CN_{I}(u, v)$ if the edge $uw$ is active at time-step $t$ and the edge $vw$ is active at time-step $t'$ for possibly distinct $t, t' \in I$. 
   
\section{Temporal $c$-Closed Graphs}\label{sec:definitions}

What structural properties could we reasonably expect to see in temporal graphs derived from social networks?  Ideally we want to identify deterministic properties of such networks, rather than rely on any of the random models of temporal social networks in the literature \cite{DBLP:journals/csur/ChakrabartiF06,takaguchi2015analyzing}, as deterministic properties are needed to design algorithms with theoretical guarantees on the worst-case running time. Moreover, as noted by \citeauthor{DBLP:journals/siamdm/KoanaKS22}~\shortcite{DBLP:journals/siamdm/KoanaKS22}, an ideal structural parameter should be computable in polynomial-time and also easy to understand.  With these considerations in mind, an excellent candidate for such a property is that of triadic closure: this formalizes the idea that people with many friends in common are likely to be friends themselves.  

This notion has been leveraged very effectively in the static setting via the notion of \emph{\cc\ graphs}.  Introduced by \citeauthor{DBLP:journals/siamcomp/FoxRSWW20}~\shortcite{DBLP:conf/icalp/FoxRSWW18,DBLP:journals/siamcomp/FoxRSWW20} as a deterministic model for social networks, \cc\ graphs are those with the property that every two vertices that have at least $c$ common neighbors are adjacent to each other.  The \emph{closure number} of a graph is the least integer $c$ for which it is $c$-closed. 

In a series of papers, \citeauthor{DBLP:journals/siamdm/KoanaKS22}~\shortcite{DBLP:journals/siamdm/KoanaKS22,DBLP:journals/algorithmica/KoanaKS23a,DBLP:journals/algorithmica/KoanaKS23} demonstrated  that the closure number and a related parameter called the weak closure number, which we will discuss shortly,  are extremely useful graph  parameters that can be exploited to design fixed-parameter tractable (\FPT) algorithms\footnote{A computational  problem is fixed-parameter tractable (\FPT) with respect to a parameter $k$  if the problem admits an algorithm that  runs in time $f(k) n^{\cO(1)}$ for some computable function $f$ (here $n$ is the input size). An algorithm with this kind of a running time is called an \FPT\ algorithm.} 
for several problems, including classic problems such as {\sc Independent Set} and {\sc Dominating Set}. These parameters have since then  received considerable attention from the (parameterized) algorithms community~\cite{DBLP:conf/innovations/BeheraH0RS22,DBLP:journals/siamdm/KaneshMRSS23,DBLP:conf/stacs/KoanaKNS22,DBLP:journals/dam/KoanaN21,DBLP:conf/fsttcs/LokshtanovS21}.

A temporal analogue of the closure number should take the timing of interactions into account; informally, we might expect that two individuals that both interact with many of the same individuals during some short time-interval are likely to interact with each other either during this same interval or shortly before or afterwards (see Figure~\ref{fig:plots} for empirical evidence).  This leads us to the following definition (in which the notions of ``short time-interval'', ``shortly before'' and ``shortly afterwards'' can be tuned by setting appropriate values of $\Delta_0$, $\Delta_1$ and $\Delta_2$).

\begin{figure}
  \centering
  \includegraphics[width=\columnwidth]{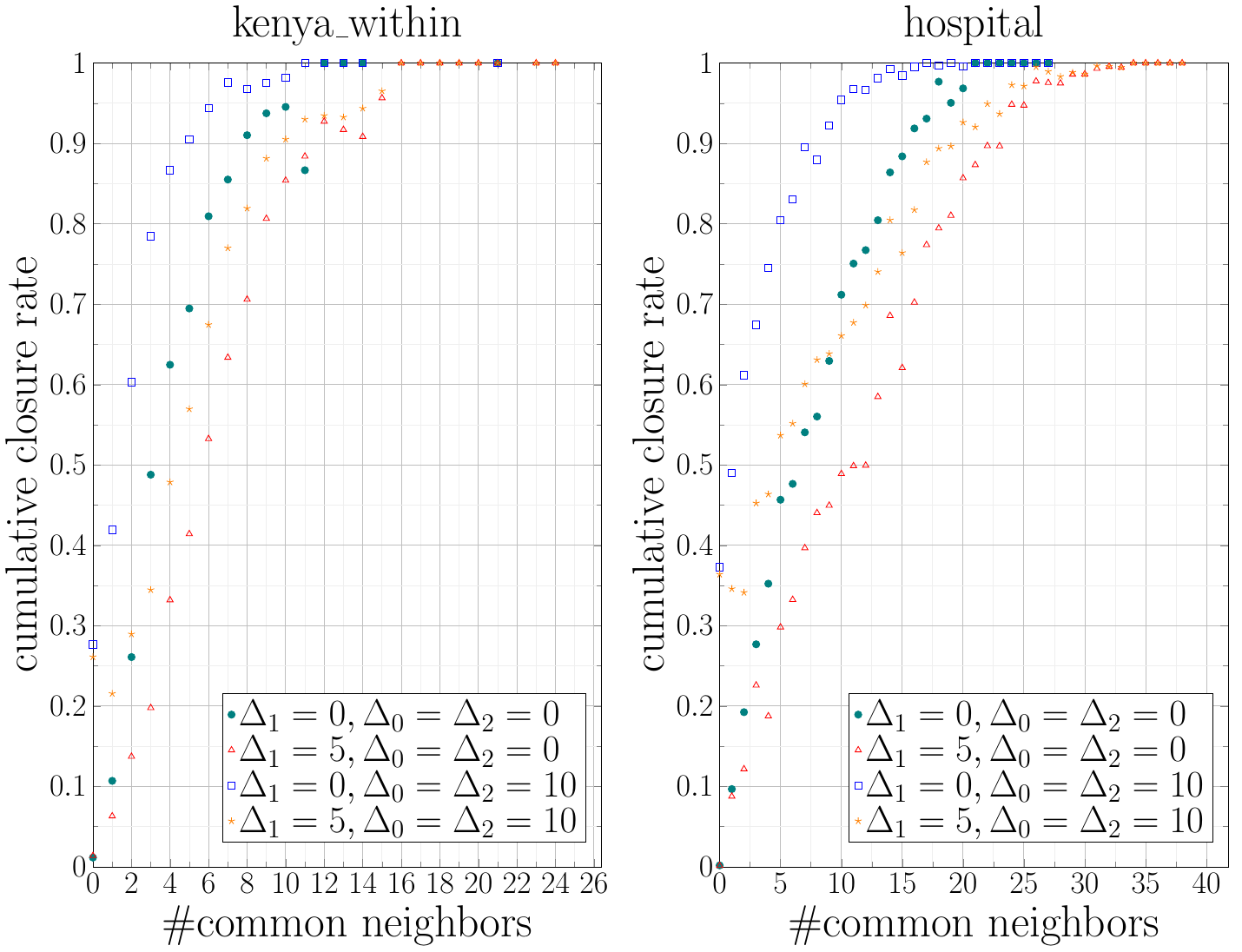}
  \caption{Cumulative closure rate of two real-world temporal networks. Each color corresponds to one  choice of \tp. For each $x$-value, the corresponding $y$-value is the cumulative closure rate, i.e., the fraction of tuples $([a,a+\Delta_1],u,v)$ such that $[a,a+\pone]\subseteq [1+\pngt,\lt - \ptwo]$ and $u$ and $v$ are distinct vertices that have $x$ common neighbors in  $[a, a+\pone]$ and are adjacent to each other during $[a - \pngt, a+\pone + \ptwo]$.}
  \label{fig:plots}
\end{figure}

    \begin{definition}[$(\pngt,\pone, \ptwo, c)$-closed graphs]\label{def:c-closed}
  For integers $\pngt, \pone, \ptwo \geq 0$ and $ c \geq 1$, a temporal graph $\ca{G} = (G, \lambda)$ is $(\pngt, \pone, \ptwo, c)$-closed if the following condition holds: for every two distinct vertices $u, v \in V(G)$ and any time-interval $[a, b]$ of length at most $\pone + 1$ (i.e., $b - a \leq \pone$) with $a \geq 1 + \pngt $ and $b \leq \lt_{\ca{G}} - \ptwo$,  if $u$ and $v$ have at least $c$ common neighbors during $[a, b]$, then $u$ and $v$ are adjacent to each other during $[a - \pngt, b + \ptwo]$.  
\end{definition}

We must note that Definition~\ref{def:c-closed} is one of the several plausible  adaptations of static \cc\ graphs to the temporal setting. But it is more general than some of the more simplistic  adaptations, say requirements such as (a) every snapshot be \cc, or (b) the footprint be \cc, or (c) the static graph during any interval of length $\pone + 1$ be \cc. By appropriately choosing $\pngt, \pone$ and $\ptwo$, we can derive each of these  three requirements as a special case of Definition~\ref{def:c-closed}. For example, fixing $\pngt = \pone = \ptwo = 0$ leads to the first requirement, i.e., every snapshot be \cc. 

    Along with \cc\ graphs, \citeauthor{DBLP:journals/siamcomp/FoxRSWW20}~\shortcite{DBLP:journals/siamcomp/FoxRSWW20} had also defined a more general class of graphs called weakly $\gamma$-closed graphs;\footnote{\citeauthor{DBLP:journals/siamcomp/FoxRSWW20}~\shortcite{DBLP:journals/siamcomp/FoxRSWW20} use the letter $c$ rather than $\gamma$ for weakly closed graphs as well. But subsequent literature on the topic~\cite{DBLP:conf/stacs/KoanaKNS22,DBLP:journals/algorithmica/KoanaKS23} use $\gamma$ for the weak version, and we defer to this trend.} we extend this too to the temporal setting. 
 For $\gamma \geq 1$, a (static)   graph $G$ is weakly $\gamma$-closed if every induced subgraph $H$ of $G$ contains a vertex $v$ such that the  number of common neighbors $v$ has with any non-neighbor (in $H$) is at most $\gamma - 1$. 

 Motivated by weakly $\gamma$-closed graphs, we define weakly \tg-closed temporal graphs, and for that, we first define the \tp-closure number of a vertex as follows. 
 Recall that $CN_{[a, b]}(u, v)$ denotes the common neighborhood of $u$ and $v$ during the time-interval $[a, b]$. 
 \begin{definition}[closure number of a vertex]\label{def:closure-vertex}
    Consider integers $\pngt, \pone, \ptwo \geq 0$ and a temporal graph $\ca{G} = (G, \lambda)$. For a vertex $v \in V(G)$, the \tp-closure of $v$ in $\ca{G}$, denoted by $\cl_{\ca{G}}(v, \tp)$, is defined as follows:
    \[
    \cl_{\ca{G}}(v, \tp) = \max_{(u, [a, b])}%
    \set{0, \card{CN_{[a, b]}(u, v)}}, %
    \]
    where the maximum is over all vertices $u$ and time-intervals $[a, b]$ such that $u \in V(G) \setminus \set{v}$, $[a, b] \subseteq [1 + \pngt, \lt_{\ca{G}} - \ptwo]$, $b - a \leq \pone$ and $uv \notin E(G_{[a - \pngt, b + \ptwo]})$. 
    That is, $\cl_{\ca{G}}(v, \tp)$ is the maximum number of common neighbors $v$ has with another vertex $u$ during any interval $[a, b]$ of length at most $\pone + 1$ such that $u$ and $v$ are not adjacent to each other during $[a - \pngt, b + \ptwo]$. 
\end{definition}

\begin{definition}[weakly \tg-closed graphs]\label{def:weak-closure}
    For integers $\pngt, \pone, \ptwo \geq 0$ and $\gamma \geq 1$, a temporal graph $\ca{G} = (G, \lambda)$ is weakly \tg-closed if one of the following two equivalent conditions holds: 
    \begin{itemize}
        \item every induced subgraph $\ca{G}'$ of $\ca{G}$ contains a vertex $v$ such that $\cl_{\ca{G}'}(v, \tp) \leq \gamma - 1$;
        
        \item there exists an ordering $v_1, v_2,\ldots, v_n$ of the vertices of $\ca{G}$ such that $\cl_{\ca{G}_i}(v_i) \leq \gamma - 1$ for every $i \in [n]$, where $\ca{G}_i$ is the subgraph of $\ca{G}$ induced by $\set{v_i, v_{i + 1},\ldots, v_n}$. 
    \end{itemize}
\end{definition} 

We define the \tp-closure number of a temporal graph $\ca{G}$ to be the least $c$  for which $\ca{G}$ is \tc-closed; we define the weak \tp-closure number analogously. 
\begin{table*}
  
  \begin{tabular}{l||rrrrr||rrrr||rrr}
\rowcolor{gray!30} \cellcolor{white}    &  \multicolumn{5}{c||}{\textbf{(A)} General statistics} & \multicolumn{4}{c||}{ \textbf{(B)} Temporal $c$ and $\gamma$} & \multicolumn{3}{c}{ \textbf{(C)} Instabilities}\\[2pt]
\hline
  Instance & $|V|$ & $|E|$ & $\lt$ & Degree & $c, \gamma$ & \multicolumn{4}{c||}{$\Delta_1/\Delta_0 (\Delta_2=\Delta_0)$} & lo-$\eta$ & \multicolumn{2}{c}{pairwise-$\eta$}  \\
    $\Delta_i$ & & & & max, min & & $0/0$ & $0/10$ & $5/0$ & $5/10$  & & $\Delta_1=0$  & $\Delta_1=5$\\
  \hline 
  \textit{baboons} & 21 & 162 & 27 & 19, 1 & 15, 4 & 8, 5 & 5, 3 & 12, 11 & 12, 9 & 11 & 8 & 7\\ 
  \textit{hospital} & 73 & 1381 & 71 & 61, 2 & 45, 25 & 20, 15 & 20, 9 & 33, 21 & 33, 18 & 26 & 24 & 36\\ 
  \textit{kenya\_across} & 21 & 54 & 45 & 14, 1 & 10, 4 & 8, 3 & 8, 3 & 8, 3 & 8, 3 & 13 & 6 & 8\\ 
  \textit{kenya\_within} & 47 & 479 & 61 & 39, 6 & 27, 11 & 11, 7 & 10, 6 & 15, 10 & 15, 10 & 19 & 14 & 14\\ 
  \textit{malawi} & 86 & 347 & 30 & 31, 1 & 10, 5 & 5, 4 & 4, 2 & 6, 5 & 6, 4 & 15 & 5 & 7\\
  \textit{workplace13} & 95 & 3915 & 275 & 93, 17 & 70, 45 & 36, 15 & 14, 7 & 41, 20 & 24, 14 & 78 & 77 & 77\\ 
  \textit{workplace15} & 217 & 4274 & 275 & 84, 1 & 41, 20 & 19, 12 & 19, 11 & 30, 16 & 30, 16 & 33 & 23 & 23\\ 
\end{tabular}
 \caption{Statistics for the instances used for numerical experiments. The $\Delta_i$ line contains the values of $\Delta_0,\Delta_1$ and $\Delta_2$, when they are relevant. \textbf{Part (A):}  The columns $|V|$, $|E|$, $c, \gamma$ and Deg max, min contain contain the number of vertices, the number of edges, the (static) closure  and weak closure numbers, and the maximum and minimum degrees of the footprint, respectively, and \lt\ is the lifetime of the temporal network. \textbf{Part (B):} Contains the closure and the weak closure numbers for different values of $\Delta_0, \Delta_1$ and $\Delta_2$; for an entry $a, b$ in a column, $a$ is the closure number (i.e., $c$) and $b$ is the weak closure number (i.e., $\gamma$). \textbf{Part (C):} The column ``lo-$\eta$'' contains the minimum value $\eta$ for which the graph is locally $\eta$-unstable. The columns ``pairwise-$\eta$'' contain the minimum value $\eta$ for which the graph is pairwise $\eta$-unstable with the following restriction; to make the computation feasible, instead of considering all intervals $[a, b]$ as required by Definition~\ref{def:pairwise-eta-unstable}, we only considered intervals $[a, a + \pone]$.} 
 \label{tab:stats}
\end{table*}

\paragraph{Empirical Results.}
We now present a dataset of real-world temporal graphs that we have used to observe how well these parameters worked in practice.
We have taken the data available on Sociopatterns\footnote{Available at \url{http://www.sociopatterns.org/datasets/}} which brings together a large number of real-world proximity networks. As we wanted to be able to calculate the values of the parameters on a personal computer, we selected only the modest-sized instances (\textit{i.e.} those with at most 250 vertices). The time-steps were also selected to reduce the computation time. The detailed list of the selected instances is the following:
\begin{compactitem}
  \item \textit{baboons}: Interactions between guinea baboons living in an enclosure of a primate center in France~\cite{gelardi2020measuring}. As time step, we chose one day.
   \item \textit{hospital}: Co-presence in a workplace in a hospital~\cite{Genois2018}. As time step, we chose one hour.
  \item \textit{kenya\_across} and \textit{kenya\_within}: Contact networks between members of households of rural Kenya~\cite{kiti2016quantifying}. As time step, we chose one hour.
  \item \textit{malawi}: Contact patterns in a village in rural Malawi~\cite{ozella2021using}. As time step, we chose one day.
  \item \textit{workplace13} and \textit{workplace15}:  Co-presence in a workplace in two different years~\cite{Genois2018}. As time step, we chose one hour.
\end{compactitem}
Some general statistics are depicted in Table~\ref{tab:stats}.

 As it was not possible to test every possible values of $\Delta_0,\Delta_1$ and $\Delta_2$, we chose to test four configurations for which $\Delta_0+\Delta_1+\Delta_2$ does not exceed 10\% of the lifetime of the longest instance (\textit{i.e.} \textit{workspace15}):
 \begin{inparaenum}[(a)]
   \item $\Delta_1=0$ and $\Delta_0=\Delta_2=0$,
   \item $\Delta_1=0$ and $\Delta_0=\Delta_2=10$,
   \item $\Delta_1=5$ and $\Delta_0=\Delta_2=0$,
   \item $\Delta_1=5$ and $\Delta_0=\Delta_2=10$.	

   \end{inparaenum}

Our empirical results indicate that for various of choices of \tp, the \tp-closure number and particularly the weak \tp-closure number are often small compared to the number of vertices and edges in the network. More telling is the fact that these values are considerably smaller than their static counterparts on the network's footprint (in which we ignore the time-steps and treat the network as static). 
See Table~\ref{tab:stats} for an overview.


\section{The Need for Neighborhood Stability}\label{sec:stability}

In the previous section we have introduced our notion of triadic closure for temporal graphs, and have provided some evidence that this definition is useful in describing structural properties of some social networks.  Our primary goal in defining this new parameter, however, is to capture realistic structural properties that can be leveraged to design efficient algorithms.

The most obvious computational problems to tackle with our new parameter are the temporal analogues of those known to admit efficient algorithms on static \cc\ graphs. 
The canonical such problem is clique enumeration (in which the goal is to output  all maximal subsets of vertices that induce complete subgraphs).   \citeauthor{DBLP:journals/siamcomp/FoxRSWW20}~\shortcite{DBLP:conf/icalp/FoxRSWW18,DBLP:journals/siamcomp/FoxRSWW20} showed that the number of maximal cliques in an $n$-vertex \cc\ graphs is at most $3^{(c- 1)/3} \cdot n^2$; this bound leads to an algorithm for enumerating all maximal cliques in time $3^{(c - 1)/3} \cdot \poly(n)$. 
We must note here that the number of maximal cliques in an arbitrary $n$-vertex graph could be as large as $3^{n/3}$~\cite{moon1965cliques} (which necessarily implies that any algorithm enumerating all maximal cliques requires time $\Omega(3^{n/3})$ in the worst case). 
Enumeration of cliques and similar dense subgraphs is also a natural problem to consider in the context of social networks, as such structures correspond to very highly connected communities within the network (a clique represents a community in which every two members interact directly).

To address these algorithmic problems on $(\pngt,\pone, \ptwo, c)$-closed temporal graphs, we first need a notion of temporal cliques. 
We use a notion of temporal clique that has been studied in the literature, particularly in the context of algorithms for enumerating maximal cliques~\cite{DBLP:journals/tcs/ViardLM16,DBLP:journals/snam/HimmelMNS17}. According to this notion,  a clique in a temporal graph is a subgraph in which every possible edge appears every so often. The formal definition is as follows. 

\begin{definition}[$\pcq$-clique~\cite{DBLP:journals/tcs/ViardLM16}]\label{def:delta-clique}
    Consider a temporal graph $\ca{G} = (G, \lambda)$. For a non-negative integer $\pcq$, a $\pcq$-clique in $\ca{G}$ is a pair $(X, [a, b])$, where $X \subseteq V(G)$ and $[a, b]$ is a time-interval, with the following properties: $b - a \geq \pcq$, and for any two distinct vertices $u, v \in X$ and for every $\tau \in [a, b - \pcq]$, the edge $uv$ is active during $[\tau, \tau + \pcq]$, i.e., there exists $t \in \lambda(uv) \cap [\tau, \tau + \pcq]$.
\end{definition}

Informally, a maximal \pcq-clique is one which is not contained in any other \pcq-clique. As a \pcq-clique has two constituent parts---a set of vertices and a time interval---the ``not contained in any other''  needs to be defined with respect to both. We define this formally now. Consider a \pcq-clique $(X, [a, b])$. We say that $(X, [a, b])$ is \emph{vertex-maximal} if there does not exist a vertex subset $X' \subseteq V(G)$ such that $X \subsetneq X'$ and $(X', [a, b])$ is a \pcq-clique. Similarly, we say that $(X, [a, b])$ is \emph{time-maximal} if there does not exist a time-interval $[a', b']$ such that $[a, b] \subsetneq [a', b']$ and $(X, [a', b'])$ is a \pcq-clique. And we say that  $(X, [a, b])$ is a \emph{maximal \pcq-clique} if it is both vertex-maximal and time-maximal. Notice that if a \pcq-clique $(X, [a, b])$ is not vertex-maximal, then $(X \cup \set{u}, [a, b])$ is a \pcq-clique for some vertex $u \notin X$. And if $(X, [a, b])$ is not time-maximal, then either $(X, [a - 1, b])$ is a \pcq-clique or $(X, [a, b + 1])$ is a \pcq-clique.

The number of maximal $\pcq$-cliques in an $n$-vertex temporal graph $\ca{G}$ could be as large as $2^{n} \cdot \lt_{\ca{G}}$ (see Example~\ref{ex:without-stability} below). \citeauthor{DBLP:journals/snam/HimmelMNS17}~\shortcite{DBLP:journals/snam/HimmelMNS17} introduced a parameter called $\pcq$-slice degeneracy, denoted by $d$, which is the maximum degeneracy of the underlying graph during any time-interval of length $\pcq$, and showed that the number of maximal \pcq-cliques is at most $3^{d/3} \cdot 2^{d + 1} \cdot n \cdot \lt$. A bound for maximal cliques w.r.t. a parameter called vertex deletion to order preservation is implicit in the work of \citeauthor{DBLP:journals/tcs/HermelinIMN23}~\shortcite{DBLP:journals/tcs/HermelinIMN23}. 

Ideally, we would like to extend the results of ~\citeauthor{DBLP:journals/siamcomp/FoxRSWW20}~\shortcite{DBLP:journals/siamcomp/FoxRSWW20} to the temporal setting; that is, bound the number of maximal \pcq-cliques in a \tc-closed temporal graph by $f(c, \pngt, \pone, \ptwo, \pcq) \cdot \poly(n, \lt)$  for some function $f$, where $n$ and $\lt$ respectively are the number of vertices and the lifetime of the temporal graph. But this is not possible: It is not difficult to construct pathological examples of temporal graphs that are \tc-closed but have $\Omega(2^n)$ maximal \pcq-cliques; see Example~\ref{ex:without-stability}. 

\begin{example}\label{ex:without-stability}
    Consider the temporal  graph $\ca{G}' = (G', \lambda')$ defined as follows. The footprint $G'$ is a clique on $n$ vertices, and we will assign time-steps to the edges in such a way that for each non-empty $X \subseteq V(G')$, we will have a maximal \pcq-clique $(X, [a, b])$ for an appropriate time-interval $[a, b]$. Let $X_1,X_2,\ldots, X_{2^n - 1}$ be the non-empty subsets of $V(G')$ (ordered arbitrarily). For $i \in [2^n - 1]$, let $t_i = (\pcq + 2)i - 1$. Now, for each $uv \in E(G')$, we define $\lambda'(uv) = \set{t_i ~|~ u, v \in X_i}$ so that each $X_i$ induces a clique at time-step $t_i$. Notice also that $t_{i + 1} - t_i =  [(\pcq + 2)(i + 1) - 1] - [(\pcq + 2)i - 1]  = \pcq + 2$, and thus there is a gap of $ \pcq + 2$ time-steps  between the cliques induced by $X_i$ and $X_{i + 1}$. Hence, for each $X_i$ with $\card{X_i} \geq 2$, $(X_i, ~ [t_i - \pcq, t_i + \pcq])$ is a maximal \pcq-clique.  As for $X_i$s  with $\card{X_i} = 1$, notice that $(X_i, [1, \lt])$ is trivially a maximal \pcq-clique. 
     It is straightforward to verify that $\ca{G'}$ is  $(\pngt, \pone, \ptwo, 1)$-closed for any $\pngt, \pone, \ptwo \geq 0$ such that $\pone \leq \pcq$; for any interval $[a, b]$ with $b - a \leq \pone \leq \pcq$, if two vertices $u$ and $v$ have at least one common neighbor during $[a, b]$, then $t_i \in [a, b]$ and $u, v \in X_i$ for some $i \in [2^n - 1]$, and hence $u$ and $v$ are adjacent to each other during $[a, b]$. 
\end{example}

Example~\ref{ex:without-stability} works because of all the cliques that appear suddenly only to disappear in an instant. Or at a local level, neighborhoods of vertices change too drastically and too suddenly. We refer to such changes as local instability, and try to limit them, which leads us to the following definition. 

\begin{definition}[locally $\eta$-unstable graphs]\label{def:eta-unstable}
    For a non-negative integer $\eta$, a temporal graph $\ca{G} = (G, \lambda)$ is locally $\eta$-unstable if $\max\set{\card{N_t(v) \setminus N_{t + 1}(v)}, \card{N_{t + 1}(v) \setminus N_{t}(v)}} \leq \eta$ for every vertex $v \in V(G)$ and every time-step $t \leq \lt - 1$. 
\end{definition}

We may think of the measure of $\eta$ as capturing the rate of local change in adjacencies. And when $\eta$ is small, the graph evolves slowly over time. In particular, when $\eta = 0$, the graph does not evolve at all; so a locally $0$-unstable temporal graph is essentially a static graph. Unfortunately, initial investigations on real datasets indicate that requiring this parameter to be small may often be too restrictive; see Table~\ref{tab:stats}.  This leads us to the following weaker (and less intuitive) version of instability, in which common neighborhoods of pairs of vertices change slowly over time, which is nevertheless sufficient for our purposes. We define this below and observe from Table~\ref{tab:stats} that in some cases it takes much smaller values than our first version.

\begin{definition}[pairwise $\eta$-unstable graphs]\label{def:pairwise-eta-unstable}
    For a non-negative integer $\eta$, a temporal graph $\ca{G} = (G, \lambda)$ is pairwise  $\eta$-unstable if for every two distinct vertices $u$ and $v$, every time-interval $[a, b]$ and every two non-negative integers $\ell, \ell'$ such that $[a - \ell, b + \ell'] \subseteq [1, \lt_{\ca{G}}]$, we have $\card{CN_{[a - \ell, b + \ell']}(u, v) \setminus CN_{[a, b]}(u, v)} \leq \eta (\ell + \ell')$.\footnote{There is a typo in the conference version of this paper: For the interval $[a, b]$ in Definition~\ref{def:pairwise-eta-unstable}, an additional constraint requiring $b - a = \pone$ was inadvertently included in the conference version. But this was a typo, and must be ignored.}   
\end{definition} 

It is not difficult to prove that if a temporal graph $\ca{G}$ is locally $\eta$-unstable, then it is also pairwise-$2\eta$-unstable;\footnote{But the converse need not hold. For example, it can the the case that two vertices $u$ and $v$ have no common neighbors during any interval. But $N(u)$ could change drastically between time-steps $t$ and $t + 1$.} if $\ca{G}$ is locally $\eta$-unstable, then for any time-interval $[a, b]$ and $\ell, \ell'$, each of the $\ell + \ell'$ time-steps in $[a - \ell', b + \ell] \setminus [a, b]$ contributes at most $2\eta$ vertices to $CN_{[a - \ell', b + \ell]}(u, v) \setminus CN_{[a, b]}(u, v)$. The same argument, when applied to a \tc-closed graph $\ca{G}$ and vertices $u$ and $v$ that are not adjacent during an interval of length at least $\pngt + \pone + \ptwo + 1$,  implies the following lemma, which  
we will use in Section~\ref{sec:bounds} to bound the number of maximal \pcq-cliques. 

\begin{lemma}\label{lem:change-new}
    Let $\ca{G}$ be a locally $\eta$-unstable, \tc-closed temporal graph. For any two distinct vertices $u, v \in V(G)$  such that $uv \notin E(G_{[a, b]})$ for some time-interval $[a, b]$ where $b - a \geq \pngt + \pone + \ptwo$, it holds that $\card{CN_{[a, b]}(u, v)} \leq c - 1 + 2\eta(b - a - \pone)$. 
\end{lemma}
\begin{proof}
    The proof of the lemma follows from the following more general result. 
    \begin{claim}\label{claim:change-in-nbd}
    Consider a locally $\eta$-unstable temporal graph $\ca{G} = (G, \lambda)$ and a pair of distinct vertices $u, v \in V(G)$. Then, for every time-interval $[a, b]$ and non-negative integers $\ell, \ell'$ such that $a - \ell \geq 1$ and $b + \ell' \leq \lt_{\ca{G}}$, it holds that 
    \begin{enumerate}
        \item $\card{N_{[a - \ell, b + \ell']}(v) \setminus N_{[a, b]}(v)} \leq \eta(\ell + \ell')$ for every $v \in V(G)$, 

        \item $\card{CN_{[a - \ell, b + \ell']}(u, v) \setminus CN_{[a, b]}(u, v)} \leq 2\eta (\ell + \ell')$ for any distinct $u, v \in V(G)$, and hence 

        \item\label{item:change-in-nbd-last} $\card{CN_{[a - \ell, b + \ell']}(u, v)} \leq \card{CN_{[a, b]}(u, v)} + 2 \eta (\ell + \ell')$ for any distinct $u, v \in V(G)$. 
    \end{enumerate}
\end{claim}
    Assuming Claim~\ref{claim:change-in-nbd}, let us first complete the proof of Lemma~\ref{lem:change-new}. 
    Consider a time-interval $[a, b]$ with $b - a \geq \pngt + \pone + \ptwo$ and two distinct vertices $u, v \in V(G)$ such that $uv \notin E(G_{[a, b]})$. As $b - a \geq \pngt + \pone + \ptwo$, we have $[a, a + \pngt + \pone + \ptwo] \subseteq [a, b]$, and hence $uv \notin E(G_{[a, a + \pngt + \pone + \ptwo]})$. But then, as $\ca{G}$ is \tc-closed, $u$ and $v$ have at most $c - 1$ common neighbors during $[a + \pngt, a + \pngt + \pone]$; that is,  $\card{CN_{[a + \pngt, a + \pngt + \pone]}(u, v)} \leq c - 1$. Claim~\ref{claim:change-in-nbd}-item~\ref{item:change-in-nbd-last} (with $\ell = \pngt$ and $\ell' = b - a - \pngt - \pone$) then implies that $\card{CN_{[a, b]}(u, v)} \leq \card{CN_{[a + \pngt, a + \pngt + \pone]}(u, v)} + 2\eta(\pngt + ((b - a) - (\pngt + \pone))) \leq c - 1 + 2 \eta(b - a - \pone)$. \\
    
    \noindent
    {\textit{Proof of Claim~\ref{claim:change-in-nbd}}.} Consider a time-interval $[a, b]$ and non-negative integers $\ell, \ell'$. 
     \begin{enumerate}
         \item Recall first that as $\ca{G}$ is locally $\eta$-unstable, we have $\card{N_{t + 1}(v) \setminus N_{t}(v)} \leq \eta$ and $\card{N_{t}(v) \setminus N_{t + 1}(v)} \leq \eta$ for every vertex $v \in V(G)$ and every time-step $t \leq \lt_{\ca{G}} - 1$. Consider $v \in V(G)$. Informally, item 1 of the claim  holds because each of the $\ell + \ell'$ time-steps in $[a - \ell', b + \ell] \setminus [a, b]$ contributes at most $\eta$ vertices to $N_{[a - \ell', b + \ell]}(v) \setminus N_{[a, b]}(v)$. We now prove this  formally. 
         Consider the sets  
         \begin{align*}
           A_v = \bigcup_{j = 1}^{\ell} (N_{a - j}(v) \setminus N_{a - j + 1}(v))
           \end{align*} and 
           \begin{align*}
           B_v = 
           \bigcup_{j' = 1}^{\ell'} (N_{b + j'}(v) \setminus N_{b + j' - 1}(v)).  
         \end{align*}
         Observe that $\card{A_v} \leq \sum_{j = 1}^{\ell} \card{(N_{a - j}(v) \setminus N_{a - j + 1}(v))} 
             \leq \sum_{j = 1}^{\ell} \eta 
             =\eta \ell$, 
                 and similarly
                $\card{B_v} \leq \eta \ell'$. \\          
         
         Let $S_v = A_v \cup B_v$. We will show that $N_{[a - \ell, b + \ell']}(v) \setminus N_{[a, b]}(v) \subseteq S_v$, 
         which will imply that $\card{N_{[a - \ell, b + \ell']}(v) \setminus N_{[a, b]}(v)}$ is at most $\card{S_v} \leq \card{A_v} + \card{B_v} \leq \eta(\ell + \ell')$. \\

         Now, consider $x \in N_{[a - \ell, b + \ell']}(v) \setminus N_{[a, b]}(v)$. Then $x$ is adjacent to $v$ during $[a - \ell, b + \ell'] \setminus [a, b]$, and $x$ is not adjacent to $v$ during $[a, b]$.  Therefore, either there exists an index $j$ with $1 \leq j \leq \ell$ such that $x \in N_{a - j}(v)$ and $x \notin N_{a - j + 1}(v)$, in which case $x \in N_{a - j}(v) \setminus N_{a - j + 1}(v)$, or there exists an index $j'$ with $1 \leq j' \leq \ell'$ such that $x \in N_{b + j'}(v)$ and $x \notin N_{b + j' - 1}(v)$, in which case $x \in N_{b + j'}(v) \setminus N_{b + j' - 1}(v)$. In any case, $x \in S_v$, and we can thus conclude that $N_{[a - \ell, b + \ell']}(v) \setminus N_{[a, b]}(v) \subseteq S_v$. 

         \item Consider distinct $u, v \in V(G)$. To prove item 2, notice that it is enough to prove that  $CN_{[a - \ell, b + \ell']}(u, v) \setminus CN_{[a, b]}(u, v) \subseteq S_u \cup S_v$, where $S_u$ and $S_v$ are as defined in item 1. 
         To that end, consider $y \in CN_{[a - \ell, b + \ell']}(u, v) \setminus CN_{[a, b]}(u, v)$. Then, $y \in CN_{[a - \ell, b + \ell']}(u, v)$, and hence $y \in N_{[a - \ell, b + \ell']}(u)$ and $y \in N_{[a - \ell, b + \ell']}(v)$. Also, $y \notin CN_{[a, b]}(u, v)$, and hence, either $y \notin N_{[a, b]}(u)$, in which case $y \in N_{[a - \ell, b + \ell']}(u) \setminus N_{[a, b]}(u) \subseteq S_u$,  or $y \notin N_{[a, b]}(v)$, in which case $y \in N_{[a - \ell, b + \ell']}(v) \setminus N_{[a, b]}(v) \subseteq S_v$. In any case, $y \in S_u \cup S_v$. We  thus have $CN_{[a - \ell, b + \ell']}(u, v) \setminus CN_{[a, b]}(u, v) \subseteq S_u \cup S_v$, and item 2 follows.  

         \item Observe that item 3 is an immediate consequence of item 2.  \qedhere 
     \end{enumerate} 
\end{proof}

\section{Bound for Maximal  Cliques}\label{sec:bounds}

We now bound the number of maximal \pcq-cliques in a \tc\ closed, locally $\eta$ unstable temporal graph. Specifically, we prove the following theorem. 
\begin{theorem}\label{thm:maximal-cliques-bound}
    For every $c \geq 1$ and $\eta, \pngt, \pone, \ptwo \geq 0$, where $\pcq \geq \pngt + \pone + \ptwo$, any $\eta$-unstable, \tc-closed temporal graph with $n$ vertices and lifetime $\lt$ has at most \mcbound\ maximal $\pcq$-cliques. 
\end{theorem}

We need the following observation to prove  Theorem~\ref{thm:maximal-cliques-bound}. 
\begin{observation}\label{obs:clique-misc}
    Consider a temporal graph $\ca{G} = (G, \lambda)$ and a non-negative integer $\pcq$. For each $X \subseteq V(G)$ and each positive integer $\tau \leq \lt_{\ca{G}}$, there exists at most one time-interval $[a, b]$ such that $\tau \in [a, b]$ and $(X, [a, b])$ is a maximal \pcq-clique.\footnote{To see why Observation~\ref{obs:clique-misc} is true, notice that if there exist two such time-intervals $[a_1, b_1]$ and $[a_2, b_2]$, then for $a = \min\set{a_1, a_2}$ and $b = \max\set{b_1, b_2}$,  $(X, [a, b])$ is also a \pcq-clique, which contradicts the time-maximality of $(X, [a_1, b_1])$ and $(X, [a_2, b_2])$.  As an aside, notice that this observation immediately implies that the number of maximal \pcq-cliques in $\ca{G}$ is at most $2^{\card{V(G)}} \cdot \lt$, as each pair $(X, \tau)$, where\ $X \subseteq V(G)$ and $1 \leq \tau \leq \lt_{G}$, corresponds to at most one maximal \pcq-clique.} Notice that  this holds for \emph{all} temporal graphs, and not just \tc-closed or locally (or pairwise) $\eta$-unstable temporal graphs. 
\end{observation}

\begin{proof}[Proof of Theorem~\ref{thm:maximal-cliques-bound}]
   Let $F(\eta, \pngt, \pone, \ptwo, c, \pcq, n, \lt)$ be the maximum number of maximal \pcq-cliques in an $\eta$-unstable \tc-closed temporal graph with $n$ vertices and lifetime $\lt$; for convenience, we use $\rho$ as a shorthand for $(\eta, \pngt, \pone, \ptwo, c)$, and write $F(\rho, \pcq, n, \lt)$. Let $\ca{G} = (G, \lambda)$ be an $\eta$-unstable, \tc-closed temporal graph with $n$ vertices and lifetime \lt, and let $v$ be an arbitrary vertex of $\ca{G}$. Then every maximal \pcq-clique $(X, [a, b])$ in $\ca{G}$ is of one (or more) of the following five types. 
    \begin{itemize}
        \item Type 1: $X$  does not contain $v$ and hence $(X, [a, b])$ is a maximal \pcq-clique in $\ca{G} - v$.

        \item Type 2: $X$ contains $v$ and $(X \setminus \set{v}, [a, b])$ is a maximal \pcq-clique in $\ca{G} - v$.

        \item Type 3: $X$ contains $v$ and $(X \setminus \set{v}, [a, b])$ is not a vertex-maximal \pcq-clique in $\ca{G} - v$. 

        \item Type 4a: $X$ contains $v$, $(X \setminus \set{v}, [a, b])$ is not a time-maximal \pcq-clique in $\ca{G} - v$, but  $(X \setminus \set{v}, [a - 1, b])$ is a \pcq-clique in $\ca{G} - v$. 

        \item Type 4b: $X$ contains $v$, $(X \setminus \set{v}, [a, b])$ is not a time-maximal \pcq-clique in $\ca{G} - v$, but $(X \setminus \set{v}, [a, b + 1])$ is a \pcq-clique in $\ca{G} - v$. 
    \end{itemize}
    We will bound the number of maximal \pcq-cliques of each type. First, types 1 and 2. As they are maximal \pcq-cliques in $\ca{G} - v$, and since  $\ca{G} - v$, being an induced subgraph of $\ca{G}$, is an $\eta$-unstable and \tc-closed temporal graph with $n - 1$ vertices and lifetime at most \lt, the number of maximal \pcq-cliques in $\ca{G} - v$ is at most $F(\rho, \pcq, n - 1, \lt)$. Therefore the number of maximal \pcq-cliques of types 1 and 2 in $\ca{G}$ is at most $F(\rho, \pcq, n - 1, \lt)$. 

    Let us now bound the number of  \pcq-cliques of type 3. Consider such a \pcq-clique $(X, [a, b])$. Then $(X \setminus \set{v}, [a, b])$ is not a vertex-maximal \pcq-clique in $\ca{G} - v$, which implies that there exists a vertex $u \in V(G) \setminus X$ such that $((X \setminus \set{v}) \cup \set{u}, [a, b])$ is a \pcq-clique. As $(X, [a, b])$ is a maximal \pcq-clique and $u \notin X$, we can conclude that there exists a time-interval $[\tau, \tau + \pcq] \subseteq [a, b]$ such that $u$ and $v$ are not adjacent to each other during $[\tau, \tau + \pcq]$, i.e., $uv \notin E(G_{[\tau, \tau + \pcq]})$.  But notice that every vertex in $X \setminus \set{v}$ is adjacent to $v$ at some time-step in $[\tau, \tau + \pcq]$ (as $(X, [a, b])$ is a \pcq-clique); that is, $X \setminus \set{v} \subseteq N_{[\tau, \tau + \pcq]}(v)$. Similarly, every vertex in $X \setminus \set{v}$ is also adjacent to $u$ at some time step in $[\tau, \tau + \pcq]$ (as $((X \setminus \set{v}) \cup \set{u}, [a, b])$ is a \pcq-clique); that is, $X \setminus \set{v} \subseteq N_{[\tau, \tau + \pcq]}(u)$. We thus have $X \setminus \set{v} \subseteq N_{[\tau, \tau + \pcq]}(v) \cap N_{[\tau, \tau + \pcq]}(u) = CN_{[\tau, \tau + \pcq]}(u, v)$. Therefore, the number of choices for $X \setminus \set{v}$ is at most $2^{\card{CN_{\tau, \tau + \pcq}(u, v)}}$. 
    Also, by Observation~\ref{obs:clique-misc}, for each subset $Y$ of $CN_{[\tau, \tau + \pcq]}(u, v)$ there exists at most one time-interval $[a', b']$ such that $[\tau, \tau + \pcq] \subseteq [a', b']$ and $(Y \cup \set{v}, [a', b'])$ is a maximal \pcq-clique in $\ca{G}$. 
    Thus, by summing over all choices for $u$ and $\tau$, we get that the number of maximal \pcq-cliques of type 3 is at most 
    $
    \sum_{(u, \tau)} 2^{\card{CN_{[\tau, \tau + \pcq]}(u, v)}}, 
    $
   where the summation is over all pairs $(u, \tau)$ such that $u \in V(G) \setminus \set{v}$, $1 \leq \tau \leq \lt - \pcq$ and $uv \notin E(G_{[\tau, \tau + \pcq]})$. Now, as $\ca{G}$ is $\eta$-unstable and \tc-closed, and as $uv \notin E(G_{[\tau, \tau + \pcq]})$ with $\pcq \geq \pngt + \pone + \ptwo$, we can apply Lemma~\ref{lem:change-new}, by which we have $\card{CN_{[\tau, \tau + \pcq]}(u, v)} \leq \nbdbound$. Then, as the pair $(u, \tau)$ has at most $n \cdot \lt$ choices, we get that the number of maximal \pcq-cliques of type 3 is at most  $2^{\nbdbound} \cdot n \cdot \lt$.

   We now bound the number of maximal \pcq-cliques of type 4a. Consider a \pcq-clique $(X, [a, b])$ of type 4a. Then $(X \setminus \set{v}, [a, b])$ is not time-maximal, and $(X \setminus \set{v}, [a - 1, b])$ is a \pcq-clique. As $(X, [a, b])$ is a maximal \pcq-clique, and in particular a time-maximal \pcq-clique, $(X, [a - 1, b])$ is not a \pcq-clique, which, along with the fact that $(X \setminus \set{v}, [a - 1, b])$ is a \pcq-clique, implies that there exists a vertex $u \in X \setminus \set{v}$ such that $u$ and $v$ are not adjacent to each other during the interval $[a - 1, a - 1 + \pcq]$. That is, $uv \notin E(G_{[a - 1, a - 1 + \pcq]})$. 
   Then, as before, 
   Lemma~\ref{lem:change-new} applies, and we thus have $\card{CN_{[a - 1, a + \pcq]}(u, v)} \leq \nbdboundplus$.  
   Now, as $(X, [a, b])$ is a \pcq-clique, and $u, v \in X$, every vertex $w \in X \setminus \set{u, v}$ is adjacent to both $u$ and $v$ during the interval $[a, a + \pcq]$. That is, $uw, vw \in E(G_{[a, a + \pcq]})$ for every $w \in X \setminus \set{u, v}$, which implies that $X \setminus \set{u, v} \subseteq CN_{[a, a + \pcq]}(u, v) \subseteq CN_{[a - 1, a + \pcq]}(u, v)$. Hence, $\card{X \setminus \set{u, v}} 
   \leq \card{CN_{a - 1, a - 1 + \pcq}(u, v)} \leq \nbdboundplus$, and thus the number of choices for $X \setminus \set{u, v}$ is at most $2^{\nbdboundplus}$. By summing over all choices of $u$ and $a$, we get that the number of maximal \pcq-cliques of type 4a is at most 
   $
   \sum_{(u, a)} 2^{\nbdboundplus}, 
   $ 
   where the summation is over all pairs $(u, a)$ such that $u \in V(G) \setminus \set{v}$, and $2 \leq a \leq \lt - \pcq$, and  $uv \notin E(G_{[a - 1, a - 1 + \pcq]})$. As $(u, a)$ has at most $n \cdot \lt$ choices, we get that the number of maximal \pcq-cliques of type 4a is at most $2^{\nbdboundplus} \cdot n \cdot \lt$. By symmetric arguments, we can also conclude that the number of maximal \pcq-cliques of type 4b is at most $2^{\nbdboundplus} \cdot n \cdot \lt$. 

   We have thus shown that the number of maximal \pcq-cliques of types 3, 4a and 4b together is at most 
   \[
   \underbrace{2^{\nbdbound} \cdot n \cdot \lt}_{\text{Type 3}} + \underbrace{2 \cdot 2^{\nbdboundplus} \cdot n \cdot \lt}_{\text{Types 4a and 4b}},
   \]
   which is at most $3 \cdot 2^{\nbdboundplus} \cdot n \cdot \lt$. Recall that the number of maximal \pcq-cliques of types 1 and 2 together is at most $F(\rho, \pcq, n - 1, \lt)$. 
   Thus, taking into account all five types, the number of maximal \pcq-cliques in $\ca{G}$ is governed by the recursive inequality 
   \[
   F(\rho, \pcq, n, \lt) \leq F(\rho, \pcq, n - 1, \lt) + 3 \cdot 2^{\nbdboundplus} \cdot n \cdot \lt. 
   \]
   By induction on $n$ with  base case $F(\rho, \pcq, 1, \lt) = 1$, we get that $F(\rho, \pcq, n, \lt)$ is at most $3 \cdot 2^{\nbdboundplus} \cdot n^2 \cdot \lt$.
 
   This completes the proof of the theorem. 
\end{proof}

\subsubsection*{Extensions and Implications of Theorem~\ref{thm:maximal-cliques-bound}.} We now note down a number of results that can be derived by using the same arguments as in the proof of Theorem~\ref{thm:maximal-cliques-bound}. Notice first that in the proof of Theorem~\ref{thm:maximal-cliques-bound}, the vertex $v$ was chosen arbitrarily, and the only property of $v$ that we used was $\cl_{\ca{G}}(v, \tp) \leq c - 1$. Hence the same proof will still go through so long as there exists a vertex $v$ with this property at each recursive level; that is, there exists an ordering $v_1, v_2,\ldots, v_n$ of the vertices of $\ca{G}$ such that $\cl_{\ca{G}_i}(v_i, \tp) \leq c - 1$, where $\ca{G}_i$ is the subgraph of $\ca{G}$ induced by $v_i, v_{i + 1},\ldots, v_n$. Recall that this is precisely the definition of weakly \tc-closed graphs (Definition~\ref{def:weak-closure}). We thus have the following result. 

\begin{theorem}\label{thm:weak-maximal-cliques-bound}
    For every $\gamma \geq 1$ and $\eta, \pngt, \pone, \ptwo \geq 0$, where $\pcq \geq \pngt + \pone + \ptwo$, every locally $\eta$-unstable, weakly \tg-closed temporal graph with $n$ vertices and lifetime $\lt$ has at most \weakmcbound\ maximal $\pcq$-cliques.
\end{theorem}

Observe also that the proof of Theorem~\ref{thm:maximal-cliques-bound} can naturally be turned into an algorithm that enumerates all maximal \pcq-cliques. But we may instead use any already known algorithm for enumerating maximal \pcq-cliques. In particular, by using an algorithm due to \citeauthor{DBLP:journals/snam/HimmelMNS17}~\shortcite{DBLP:journals/snam/HimmelMNS17},\footnote{We use Theorem 2 in \cite{DBLP:journals/snam/HimmelMNS17}, which says that all maximal \pcq-cliques can be enumerated in time $\cO(x \cdot m + m \cdot \lt)$, where $x$ is the number of \emph{time-maximal} \pcq-cliques,  and $m$ is the number of edges in the footprint. To use this result, we therefore need to bound the number of time-maximal \pcq-cliques. 
But it is easy to adapt our proof of Theorem~\ref{thm:maximal-cliques-bound} to bound the number of time-maximal \pcq-cliques by $2 \cdot 2^{\nbdboundplus} \cdot n^2 \cdot \lt$; we only have to omit Type 3 cliques.} we get the following result. 

\begin{theorem}\label{thm:clique-algorithm}
    There is an algorithm that given a locally $\eta$-unstable, \tc-closed temporal graph $\ca{G} = (G, \lambda)$ and $\pcq \geq \pngt + \pone + \ptwo$, runs in time $\cO(2^{\nbdboundplus} \cdot n^2 \cdot m \cdot \lt)$, and enumerates all maximal \pcq-cliques in $\ca{G}$, where $n$ and $m$ respectively are the number of vertices and edges of the footprint $G$.  
\end{theorem}

 Theorem~\ref{thm:clique-algorithm} says that the maximal clique enumeration problem and consequently the decision problem of checking if a temporal graph contains a clique of a given size are fixed-parameter tractable, when parameterized by $c + \eta + \pcq$.  
 
\begin{remark}[Pairwise $\eta$-instability is sufficient] 
     In the proof of Theorem~\ref{thm:maximal-cliques-bound}, we used the local $\eta$-instability of $\ca{G}$ only when we invoked Lemma~\ref{lem:change-new}; for example, in the type 3 case, to bound $\card{CN_{[\tau, \tau + \pcq]}(u, v)}$ using the fact that $\card{CN_{[\tau + \pngt, \tau + \pngt + \pone]}(u, v)} \leq c - 1$.  For this, notice that pairwise $\eta$-instability is sufficient. Thus we may replace local $\eta$-instability with pairwise $2\eta$-instability in Theorems~\ref{thm:maximal-cliques-bound},  \ref{thm:weak-maximal-cliques-bound} and \ref{thm:clique-algorithm}, and the same bounds will still hold; the same goes for Theorem~\ref{thm:maximal-plexes-bound} in Section~\ref{sec:bounds-more}. 
\end{remark}

\begin{remark}[Necessity of the assumptions in Theorem~\ref{thm:maximal-cliques-bound}]
    Theorem~\ref{thm:maximal-cliques-bound} crucially relies on three  requirements: $\eta$-instability, \tc-closure and $\pcq \geq \pngt + \pone + \ptwo$. Each of these  requirements is necessary to yield our bound for the number of maximal \pcq-cliques. Without any one of them, the number of maximal \pcq-cliques may blow up to $2^{\Omega(n)}$. We already saw the necessity of $\eta$-instability in Example~\ref{ex:without-stability}. The necessity of \tc\ closure is even more straightforward, because for every $n \geq 1$, there exists a static graph with $3n$ vertices and $3^n$ maximal cliques; this is the classic Moon-Moser graph, the complete $n$-partite graph with exactly 3 vertices in each part~\cite{moon1965cliques}. To adapt it to the temporal setting, we only need to  assign all time-steps from $1$ to $\pcq + 1$ to every edge; the resulting temporal graph is $0$-unstable, has lifetime $\pcq + 1$ and contains $3^n$ maximal \pcq-cliques, but not \tc-closed any $\pngt, \pone, \ptwo \geq 0$ with $\pngt + \pone + \ptwo \leq \pcq$ and $c \leq 3n - 3$. Finally, to see that $\pcq \geq \pngt + \pone + \ptwo$ is also necessary, notice that for any $c \geq 1$ and any temporal graph $\ca{G}$ (including the temporal adaptation of the Moon-Moser graph that we just saw), we can always choose \tp\ in such a way that $\ca{G}$ is \tc-closed; for example, if we choose $\pngt = \ptwo = \lt_{\ca{G}}$, then $\ca{G}$ would be vacuously \tc\ closed for any $c$. %
\end{remark}

\setlength{\tabcolsep}{1mm}
\begin{table}
  \caption{Weakly versions of closure and pairwise instability parameters. Columns $\gamma$ are weak closure numbers, columns $\eta$ are weak  pairwise instability and columns ``b'' are the ``weak version'' that minimizes both closure and pairwise instability simultaneously.}
  \begin{tabular}{l||rrr|rrr|rrr|rrr}
    & \multicolumn{12}{c}{$\Delta_1/\Delta_2$}\\
    Instance & \multicolumn{3}{c|}{$0/0$}
    & \multicolumn{3}{c|}{$0/10$}
    & \multicolumn{3}{c|}{$5/0$}
    & \multicolumn{3}{c}{$5/10$}
    \\
    & $\gamma$  & $\eta$ & b & $\gamma$ & $\eta$ & b & $\gamma$ & $\eta$ & b & $\gamma$ & $\eta$ & b \\
    \hline 
    \textit{baboons} & 5 & 6 & 6 & 3 & 6 & 6 & 11 & 5 & 11 & 9 & 5 & 9\\ 
    \textit{hospital} & 15 & 18 & 18 & 9 & 18 & 18 & 21 & 22 & 25 & 18 & 22 & 24\\ 
    \textit{kenya\_acr} & 3 & 3 & 3 & 3 & 3 & 3 & 3 & 3 & 3 & 3 & 3 & 3\\ 
    \textit{kenya\_wit} & 7 & 9 & 9 & 6 & 9 & 9 & 10 & 8 & 10 & 10 & 8 & 10\\ 
    \textit{malawi} & 4 & 4 & 4 & 2 & 4 & 4 & 5 & 4 & 5 & 4 & 4 & 5\\
    \textit{workp13} & 15 & 77 & 77 & 7 & 77 & 77 & 20 & 77 & 77 & 14 & 77 & 77\\ 
    \textit{work15} & 12 & 14 & 15 & 11 & 14 & 15 & 16 & 14 & 16 & 16 & 14 & 16\\ 
  \end{tabular}
  \label{tab:weakly}
\end{table}

\begin{remark}[Weak versions of pairwise-instability]
    Just like the weak $\gamma$-closure, it is possible to define a weak version of the pairwise instability. 
    A temporal graph is weakly pairwise $\eta$-unstable if it is possible to find an ordering $v_1,\dots,v_n$ of the vertices such that for pair $v_i$ and $v_j$ with $i < j$, any time interval $[a,b]$ with $b-a = \Delta_1$ and any two non-negative integers $\ell,\ell'$ such that $[a - \ell,b+\ell'] \subseteq [1,\lt]$, we have $\card{CN_{[a - \ell, b + \ell']}(u, v) \setminus CN_{[a, b]}(u, v)} \leq \eta (\ell + \ell')$ in the subgraph induced by $\{v_i,\dots,v_n\}$. 
Another particularly interesting ``weak version'' consists in finding an order of the vertices that minimizes both $\gamma$ and $\eta$. That is, an ordering $v_1,\dots,v_n$ that minimizes the maximum between the temporal closure of $v_i$ and the pairwise instability of $v_i$ in the graph induced by $\{v_i,\dots,v_n\}$. As we can observe in Table~\ref{tab:weakly}, the weakly versions have lower values than their non-weakly counterparts.
\end{remark} 

\section{Bounds for Other Dense Subgraphs}\label{sec:bounds-more}
In Section~\ref{sec:bounds}, we bounded the number of maximal \pcq-cliques. Using the same ideas, we can prove similar bounds for other dense subgraphs such as the temporal adaptations of $(k + 1)$-plexes and $k$-defective cliques. For $k \geq 0$, a $(k + 1)$-plex is a static graph in which every vertex has at most $k$ non-neighbors, and a $k$-defective clique is one which has at least $\binom{n}{2} - k$ edges (where $n$ is the number of vertices). Thus a static clique is a $1$-plex and a $0$-defective clique.

   \citeauthor{DBLP:conf/innovations/BeheraH0RS22}~\shortcite{DBLP:conf/innovations/BeheraH0RS22} showed that various dense subgraphs also admit bounds of the form $f(c) \cdot \poly(n)$ in a static \cc\ graph. In particular, they showed that the number of maximal $(k + 1)$-plexes in an $n$-vertex \cc\ graph is at most $n^{2k} r_k^{c} c^{\cO(1)}$, where $r_k < 2$ is a constant. The other dense subgraphs they considered were complements of  forests, bounded treewidth graphs and bounded degeneracy graphs. \citeauthor{DBLP:journals/algorithmica/KoanaKS23a}~\shortcite{DBLP:journals/algorithmica/KoanaKS23a} proved  similar bounds on weakly $\gamma$-closed graphs; they showed that in an $n$-vertex weakly $\gamma$-closed graph, the number of maximal $(k + 1)$-plexes is at most $2^{\gamma} n^{2k - 1}$ and the number of maximal $k$-defective cliques is at most $2^{\gamma} n^{k + 1}$. 

\citeauthor{DBLP:journals/jea/BentertHMMNS19}~\shortcite{DBLP:journals/jea/BentertHMMNS19} adapted the definition of $(k + 1)$-plexes to the temporal setting as follows. For non-negative integers $\pcq$ and $k$, a $\ppx$-plex in a temporal graph $\ca{G}$ is a pair $(X, [a, b])$, where $X \subseteq V(G)$ and $[a, b]$ is a time-interval such that $b - a \geq \pcq$, and for every vertex $v \in X$ and for every $\tau \in [a, b - \pcq]$, there exist at most $k$ vertices $u \in X \setminus \set{v}$ such that $uv \notin E(G_{[\tau, \tau + \pcq]})$.  We similarly adapt $k$-defective cliques as follows. A $\pdc$-defective clique in $\ca{G}$ is a pair $(X, [a, b])$, where $X \subseteq V(G)$ and $[a, b]$ is a time-interval such that $b - a \geq \pcq$, and for every $\tau \in [a, b - \pcq]$, the static graph $G[X]$ has at least $\binom{\card{X}}{2} - k$ edges that are active during the interval $[\tau, \tau + \pcq]$. The definitions of  a maximal \ppx-plex and a maximal \pdc-defective cliques are analogous to the definition of a maximal \pcq-clique. And we prove the following result.  

\begin{theorem}\label{thm:maximal-plexes-bound}
    For $\gamma \geq 1$ and $\eta, \pngt, \pone, \ptwo, k \geq 0$, where $\pcq \geq \pngt + \pone + \ptwo$, any $\eta$-unstable, weakly \tg-closed temporal graph with $n$ vertices and lifetime $\lt$ has at most \weakmplexbound\ maximal \ppx-plexes, and at most \weakmdcbound\ maximal \pdc-defective cliques.  
\end{theorem}
\begin{proof}
    Let us first bound the number of maximal \ppx-plexes in a \tc-closed graph. Consider $c \geq 1$ and $\eta, \pngt, \pone, \ptwo \geq 0$, where $\pcq \geq \pngt + \pone + \ptwo$. Let $F(\eta, \pngt, \pone, \ptwo, c, \pcq, k, n, \lt)$ be the maximum number of maximal \ppx-plexes in an $\eta$-unstable, \tc-closed temporal graph with $n$ vertices and lifetime $\lt$; as before, we use $\rho$ as a shorthand for $(\eta, \pngt, \pone, \ptwo, c)$, and write $F(\rho, \pcq, k, n, \lt)$. 

    Let $\ca{G} = (G, \lambda)$ be an $\eta$-unstable, \tc-closed temporal graph with $n$ vertices and lifetime \lt, and let $v$ be an arbitrary vertex of $\ca{G}$. Then every maximal \ppx-plex $(X, [a, b])$ in $\ca{G}$ is of one (or more) of the following four  types. 
    \begin{itemize}
        \item Type 1: $X$ does not contain $v$ and hence $(X, [a, b])$ is a maximal \ppx-plex in $\ca{G} - v$.

        \item Type 2: $X$ contains $v$ and $(X \setminus \set{v}, [a, b])$ is a maximal \ppx-plex in $\ca{G} - v$.

        \item Type 3: $X$ contains $v$, $(X \setminus \set{v}, [a, b])$ is not a maximal \ppx-plex in $\ca{G} - v$, and there exist a vertex $u \in X \setminus \set{v}$ and a time-step $\tau \in [a, b - \pcq]$ such that $uv \notin E(G_{[\tau, \tau + \pcq]})$.  

        \item Type 4: $X$ contains $v$, $(X \setminus \set{v}, [a, b])$ is not a maximal \ppx-plex in $\ca{G} - v$, and $wv \in E(G_{[\tau, \tau + \pcq]})$ for every vertex $w \in X \setminus \set{v}$ and every time-step $\tau \in [a, b - \pcq]$. We now further classify Type 4 \ppx-plexes into three sub-types as follows. 
        \begin{itemize}
            \item Type 4a: $(X \setminus \set{v}, [a, b])$ is not vertex-maximal. 
            \item Type 4b: $(X \setminus \set{v}, [a, b])$ is not time-maximal, and $(X \setminus \set{v}, [a - 1, b])$ is a \ppx-plex in $\ca{G} - v$. 
            \item Type 4c: $(X \setminus \set{v}, [a, b])$ is not time-maximal, and $(X \setminus \set{v}, [a, b + 1])$ is a \ppx-plex in $\ca{G} - v$. 
        \end{itemize}
        \end{itemize}
    We will bound the number of maximal \ppx-plexes of each type. As maximal \ppx-plexes of types 1 and 2 are maximal in $\ca{G}- v$, their number is at most $F(\rho, \pcq, k, n - 1, \lt)$. 
    
    Let us now bound the number of  \ppx-plexes of type 3. Consider such a \ppx-plex  $(X, [a, b])$. Then there exist $u \in X \setminus \set{v}$ and $\tau \in [a, b - \pcq]$ such that $uv \notin E(G_{[\tau, \tau + \pcq]})$. We partition $X \setminus \set{u, v}$ into two (possibly empty) sets $S_{X, uv}$ and $S'_{X, uv}$ as follows: $S_{X, uv}$ is the set of  vertices in $X \setminus \set{u, v}$ that are adjacent to both $u$ and $v$ during $[\tau, \tau + \pcq]$, i.e., $S_{X, uv} = (X \setminus \set{u, v}) \cap CN_{[\tau, \tau + \pcq]}(u, v)$, and $S'_{X, uv}$ is the set of vertices in $X \setminus \set{u, v}$ that are not adjacent to at least one of $u$ and $v$ during  $[\tau, \tau + \pcq]$, i.e., $S'_{X, uv} = (X \setminus \set{u, v}) \setminus CN_{[\tau, \tau + \pcq]}(u, v)$. And we will bound the number of vertices in each of the two sets, (which will in turn bound the number of possible choices for the two sets, and consequently bound the number of possible choices for $X$). 
    
    First, as $\pcq \geq \pngt + \pone + \ptwo$ and as $uv \notin E(G_{[\tau, \tau + \pcq]})$, by Lemma~\ref{lem:change-new}, we have $\card{CN_{[\tau, \tau + \pcq]}(u, v)} \leq \nbdbound$. Then, as $S_{X, uv} \subseteq CN_{[\tau, \tau + \pcq]}(u, v)$, we can conclude that the number of choices for $S_{X, uv}$ is at most $2^{\nbdbound}$. 
    Second, as $(X, [a, b])$ is a \ppx-plex and $u$ and $v$ are not adjacent to each other during   $[\tau, \tau + \pcq]$, there exist at most $k - 1$ vertices in $X \setminus \set{u, v}$ that are not adjacent to $v$ during $[\tau, \tau + \pcq]$, and at most $k - 1$ vertices in $X \setminus \set{u, v}$ that are not adjacent to $u$ during $[\tau, \tau + \pcq]$. Hence, there are at most $2(k - 1)$ vertices in $X \setminus \set{u, v}$ that are not adjacent to at least one of $u$ and $v$ during  $[\tau, \tau + \pcq]$; that is, $\card{S'_{X, uv}} \leq 2(k - 1)$. Therefore, the number of choices for $S'_{X, uv}$ is $\binom{n - 2}{2(k - 1)} \leq n^{2k - 2}$. Taken together, the number of choices for $X \setminus \set{u, v}$ is at most $2^{\nbdbound} \cdot n^{2k - 2}$. By summing over all possible choices\footnote{Here we also rely on the fact that for every $X$ and every $\tau$, there exists at most one time-interval $[a, b]$ such that $\tau \in [a, b]$ and $(X, [a, b])$ is a maximal \ppx-plex; the reasoning in Observation~\ref{obs:clique-misc} applies to \ppx-plexes as well.} for $u$ and $\tau$, we get that the number of maximal \ppx-plexes of type 3 is at most $\sum_{(u, \tau)} 2^{\nbdbound} \cdot n^{2k - 2} = 2^{\nbdbound} \cdot n^{2k - 2} \cdot n \cdot \lt$. 
    
    We now deal with maximal \ppx-plexes of type 4a. Consider such a \ppx-plex $(X, [a, b])$. Then $(X \setminus \set{v}, [a, b])$ is not vertex-maximal, and hence there exists a vertex $u \in V(G) \setminus X$ such that $((X \setminus \set{v}) \cup \set{u}, [a, b])$ is a \ppx-plex. Recall that by the definition of Type 4 \ppx-plexes, $v$ is adjacent to every vertex $w \in X \setminus \set{v}$ during every interval $[\tau, \tau +  \pcq] \subseteq [a, b]$. We can therefore conclude that there exists an interval $[\tau, \tau + \pcq]$ such that $uv \notin E(G_{[\tau, \tau + \pcq]})$, for otherwise $(X \cup \set{u}, [a, b])$ would be a \ppx-plex, a contradiction to the maximality of $(X, [a, b])$. Using arguments that are identical to those for the case of type 3, we can conclude that the number of choices for $X$ is at most $2^{\nbdbound} \cdot n^{k}$; the number of vertices in $X \setminus \set{v}$ that are not adjacent to $u$ during $[\tau, \tau + \pcq]$ is at most $k$, and the remaining vertices in $X \setminus \set{v}$ are adjacent to both $u$ and $v$ during $[\tau, \tau + \pcq]$. Again, by summing over all possible choices for $u$ and $\tau)$, the number of maximal \ppx-plexes of type 4a is at most $2^{\nbdbound} \cdot n^{k} \cdot n \cdot \lt$. 

    Now, type 4b. Consider such a \ppx-plex $(X, [a, b])$. Then $(X \setminus \set{v}, [a, b])$ is not time-maximal, and in particular, $(X \setminus \set{v}, [a - 1, b])$ is a \ppx-plex. Now, as $(X, [a - 1, b])$ is not a \ppx-plex, we can conclude that there exists a vertex $u \in X \setminus \set{v}$ such that $uv \notin E(G_{[a - 1, a - 1 + \pcq]})$. Then, by Lemma~\ref{lem:change-new}, $\card{CN_{[a - 1, a + \pcq]}} \leq \nbdboundplus$. Again, using nearly identical arguments as in the case of type 4a, the number of choices for $X$ is at most $2^{\nbdboundplus} \cdot n^k$; the number of vertices in $X \setminus \set{v}$ that are not adjacent to $u$ during $[a, a + \pcq]$ is at most $k$, and the remaining vertices in $X \setminus \set{v}$ are adjacent to both $u$ and $v$ during $[a, a + \pcq] \subseteq [a - 1, a + \pcq]$, and hence their number is at most $\card{CN_{[a - 1, a + \pcq]}(u, v)} \leq \nbdboundplus$. By summing over all possible choices for $u$ and $a$, the number of maximal \ppx-cliques of type 4b is at most $2^{\nbdbound} \cdot n^{k} \cdot n \cdot \lt$. The case of type 4c is symmetric, and thus the number of such maximal \ppx-plexes is also at most $2^{\nbdbound} \cdot n^{k} \cdot n \cdot \lt$. 

    Thus the number maximal \ppx-plexes of Types 3, 4a, 4b and 4c is at most 
    \begin{align*}
        {}& \underbrace{2^{\nbdbound} \cdot n^{2k - 2} \cdot n \cdot \lt}_{\text{Type 3}} ~~~~ + \\ 
        {}& ~~~~~~~~~~~~~~ \underbrace{3 \cdot 2^{\nbdboundplus} \cdot n^{k} \cdot n \cdot \lt}_{\text{Types 4a, 4b and 4c}},
    \end{align*}
    which is at most $4 \cdot 2^{\nbdboundplus} \cdot n^{\max\set{2k - 2, k}} \cdot n \cdot \lt$. Therefore, the number of maximal \ppx-plexes in $\ca{G}$ is at most 
    \begin{align*}
        F(\rho, \pcq, k, n, \lt) &\leq F(\rho, \pcq, k, n - 1, \lt) ~~~ + \\
        {}& ~~~~~~ 4 \cdot 2^{\nbdboundplus} \cdot n^{\max\set{2k - 2, k}} \cdot n \cdot \lt \\
        &\leq 4 \cdot 2^{\nbdboundplus} \cdot n^{\max\set{2k - 2, k}} \cdot n^2 \cdot \lt \\
        &= \mplexbound. 
    \end{align*}

    As in the case of Theorem~\ref{thm:weak-maximal-cliques-bound}, the same arguments would work even if replace $\tc$-closure with weak \tg-closure. We can thus conclude that the number of maximal \ppx-plexes is bounded by \weakmplexbound. 

    Let us now consider \pdc-defective cliques. The proof is identical to  the case of \ppx-plexes, with one minor difference. When bounding the number of type 3 \pdc-defective cliques, the number of vertices that are not adjacent to at least one of $u$ and $v$ is at most $k - 1$,  i.e., $\card{S'_{X, uv}} \leq k - 1$, (rather than $2(k - 1)$, as was the case for \ppx-plexes). And hence  the number of maximal \pdc-defective cliques of type 3 is at most $2^{\nbdbound} \cdot n^{k - 1} \cdot n \cdot \lt$. We can thus conclude that any $\eta$-unstable, \tc-closed temporal graph with $n$ vertices and lifetime $\lt$ has at most \mdcbound\ maximal \pdc-defective cliques. Again, replacing \tc-closure with weak \tg-closure, we get that any $\eta$-unstable, weakly \tg-closed temporal graph with $n$ vertices and lifetime $\lt$ has at most \weakmdcbound\ maximal \pdc-defective cliques. 
    
    This completes the proof of the theorem. 
\end{proof}

\section{Concluding Remarks}

We introduced the definition of \tc-closed temporal graphs, which formalizes the triadic closure property of social networks. Our empirical results suggest that \tc-closure number and weak \tg-closure number could be meaningful parameters in the study of real-world networks. But further evaluation of large real-world temporal networks is necessary to confirm this. Our theoretical results demonstrate the usefulness of these parameters in designing algorithms with provable worst-case guarantees. We hope that this parameter could be further exploited, and that \tc-closure number will prove to be just as useful as its static counterpart in designing algorithms for a variety of problems.  
While local (or pairwise) $\eta$-instability is one sufficient condition that, when coupled with \tc-closure, can yield a non-trivial bound for the number of maximal \pcq-cliques, it will be worth investigating whether there are more  reasonable conditions that can yield a similar bound. We must also add a note of caution here. These parameters are all rather crude abstractions of properties exhibited by real-world networks; they are based on how adjacencies behave and evolve over time in an idealized temporal social  network. They can nonetheless illuminate the structure of real-world networks. Also, less-than-realistic abstraction is often the price we must pay for algorithms with provable worst-case guarantees. 
It is not our case that these parameters will be sufficiently small for practical purposes for all real-world temporal  networks. Practical applications will  require further refinement of these parameters, and we hope that our work will trigger such inquiries. 
\section*{Acknowledgments}
This work was done while Tom Davot was at the University of Glasgow, supported by EPSRC grant EP/T004878/1. Jessica Enright is supported by EPSRC grant EP/T004878/1. Jayakrishnan Madathil and Kitty Meeks are supported by EPSRC grants EP/T004878/1 and  EP/V032305/1. 

\bibliography{arxiv-refs-condensed}

\begin{thebibliography}{29}
\providecommand{\natexlab}[1]{#1}

\bibitem[{Akrida et~al.(2021)Akrida, Mertzios, Spirakis, and
  Raptopoulos}]{DBLP:journals/jcss/AkridaMSR21}
Akrida, E.~C.; Mertzios, G.~B.; Spirakis, P.~G.; and Raptopoulos, C.~L. 2021.
\newblock The temporal explorer who returns to the base.
\newblock \emph{J. Comput. Syst. Sci.}, 120: 179--193.

\bibitem[{Behera et~al.(2022)Behera, Husic, Jain, Roughgarden, and
  Seshadhri}]{DBLP:conf/innovations/BeheraH0RS22}
Behera, B.; Husic, E.; Jain, S.; Roughgarden, T.; and Seshadhri, C. 2022.
\newblock {FPT} Algorithms for Finding Near-Cliques in c-Closed Graphs.
\newblock In \emph{{ITCS}}, volume 215 of \emph{LIPIcs}, 17:1--17:24. Schloss
  Dagstuhl - Leibniz-Zentrum f{\"{u}}r Informatik.

\bibitem[{Bentert et~al.(2019)Bentert, Himmel, Molter, Morik, Niedermeier, and
  Saitenmacher}]{DBLP:journals/jea/BentertHMMNS19}
Bentert, M.; Himmel, A.; Molter, H.; Morik, M.; Niedermeier, R.; and
  Saitenmacher, R. 2019.
\newblock Listing All Maximal \emph{k}-Plexes in Temporal Graphs.
\newblock \emph{{ACM} J. Exp. Algorithmics}, 24(1): 1.13:1--1.13:27.

\bibitem[{Bumpus and Meeks(2023)}]{DBLP:journals/algorithmica/BumpusM23}
Bumpus, B.~M.; and Meeks, K. 2023.
\newblock Edge Exploration of Temporal Graphs.
\newblock \emph{Algorithmica}, 85(3): 688--716.

\bibitem[{Casteigts et~al.(2021)Casteigts, Himmel, Molter, and
  Zschoche}]{DBLP:journals/algorithmica/CasteigtsHMZ21}
Casteigts, A.; Himmel, A.; Molter, H.; and Zschoche, P. 2021.
\newblock Finding Temporal Paths Under Waiting Time Constraints.
\newblock \emph{Algorithmica}, 83(9): 2754--2802.

\bibitem[{Chakrabarti and Faloutsos(2006)}]{DBLP:journals/csur/ChakrabartiF06}
Chakrabarti, D.; and Faloutsos, C. 2006.
\newblock Graph mining: Laws, generators, and algorithms.
\newblock \emph{{ACM} Comput. Surv.}, 38(1): 2.

\bibitem[{Enright et~al.(2024)Enright, Hand, Larios{-}Jones, and
  Meeks}]{DBLP:journals/corr/abs-2404-19453}
Enright, J.~A.; Hand, S.~D.; Larios{-}Jones, L.; and Meeks, K. 2024.
\newblock Structural Parameters for Dense Temporal Graphs.
\newblock \emph{CoRR}, abs/2404.19453.

\bibitem[{Fluschnik et~al.(2020)Fluschnik, Molter, Niedermeier, Renken, and
  Zschoche}]{DBLP:conf/birthday/FluschnikMNRZ20}
Fluschnik, T.; Molter, H.; Niedermeier, R.; Renken, M.; and Zschoche, P. 2020.
\newblock As Time Goes By: Reflections on Treewidth for Temporal Graphs.
\newblock In \emph{Treewidth, Kernels, and Algorithms}, volume 12160 of
  \emph{Lecture Notes in Computer Science}, 49--77. Springer.

\bibitem[{Fox et~al.(2018)Fox, Roughgarden, Seshadhri, Wei, and
  Wein}]{DBLP:conf/icalp/FoxRSWW18}
Fox, J.; Roughgarden, T.; Seshadhri, C.; Wei, F.; and Wein, N. 2018.
\newblock Finding Cliques in Social Networks: {A} New Distribution-Free Model.
\newblock In \emph{{ICALP}}, volume 107 of \emph{LIPIcs}, 55:1--55:15. Schloss
  Dagstuhl - Leibniz-Zentrum f{\"{u}}r Informatik.

\bibitem[{Fox et~al.(2020)Fox, Roughgarden, Seshadhri, Wei, and
  Wein}]{DBLP:journals/siamcomp/FoxRSWW20}
Fox, J.; Roughgarden, T.; Seshadhri, C.; Wei, F.; and Wein, N. 2020.
\newblock Finding Cliques in Social Networks: {A} New Distribution-Free Model.
\newblock \emph{{SIAM} J. Comput.}, 49(2): 448--464.

\bibitem[{Gelardi et~al.(2020)Gelardi, Godard, Paleressompoulle, Claidiere, and
  Barrat}]{gelardi2020measuring}
Gelardi, V.; Godard, J.; Paleressompoulle, D.; Claidiere, N.; and Barrat, A.
  2020.
\newblock Measuring social networks in primates: wearable sensors versus direct
  observations.
\newblock \emph{Proceedings of the Royal Society A}, 476(2236): 20190737.

\bibitem[{G{\'e}nois and Barrat(2018)}]{Genois2018}
G{\'e}nois, M.; and Barrat, A. 2018.
\newblock Can co-location be used as a proxy for face-to-face contacts?
\newblock \emph{EPJ Data Science}, 7(1): 11.

\bibitem[{Haag et~al.(2022)Haag, Molter, Niedermeier, and
  Renken}]{DBLP:journals/dam/HaagMNR22}
Haag, R.; Molter, H.; Niedermeier, R.; and Renken, M. 2022.
\newblock Feedback edge sets in temporal graphs.
\newblock \emph{Discret. Appl. Math.}, 307: 65--78.

\bibitem[{Hermelin et~al.(2023)Hermelin, Itzhaki, Molter, and
  Niedermeier}]{DBLP:journals/tcs/HermelinIMN23}
Hermelin, D.; Itzhaki, Y.; Molter, H.; and Niedermeier, R. 2023.
\newblock Temporal interval cliques and independent sets.
\newblock \emph{Theor. Comput. Sci.}, 961: 113885.

\bibitem[{Himmel et~al.(2017)Himmel, Molter, Niedermeier, and
  Sorge}]{DBLP:journals/snam/HimmelMNS17}
Himmel, A.; Molter, H.; Niedermeier, R.; and Sorge, M. 2017.
\newblock Adapting the {B}ron-{K}erbosch algorithm for enumerating maximal
  cliques in temporal graphs.
\newblock \emph{Soc. Netw. Anal. Min.}, 7(1): 35:1--35:16.

\bibitem[{Kanesh et~al.(2023)Kanesh, Madathil, Roy, Sahu, and
  Saurabh}]{DBLP:journals/siamdm/KaneshMRSS23}
Kanesh, L.; Madathil, J.; Roy, S.; Sahu, A.; and Saurabh, S. 2023.
\newblock Further Exploiting \emph{c}-Closure for {FPT} Algorithms and Kernels
  for Domination Problems.
\newblock \emph{{SIAM} J. Discret. Math.}, 37(4): 2626--2669.

\bibitem[{Kempe, Kleinberg, and Kumar(2002)}]{DBLP:journals/jcss/KempeKK02}
Kempe, D.; Kleinberg, J.~M.; and Kumar, A. 2002.
\newblock Connectivity and Inference Problems for Temporal Networks.
\newblock \emph{J. Comput. Syst. Sci.}, 64(4): 820--842.

\bibitem[{Kiti et~al.(2016)Kiti, Tizzoni, Kinyanjui, Koech, Munywoki, Meriac,
  Cappa, Panisson, Barrat, Cattuto et~al.}]{kiti2016quantifying}
Kiti, M.~C.; Tizzoni, M.; Kinyanjui, T.~M.; Koech, D.~C.; Munywoki, P.~K.;
  Meriac, M.; Cappa, L.; Panisson, A.; Barrat, A.; Cattuto, C.; et~al. 2016.
\newblock Quantifying social contacts in a household setting of rural Kenya
  using wearable proximity sensors.
\newblock \emph{EPJ data science}, 5: 1--21.

\bibitem[{Koana et~al.(2022)Koana, Komusiewicz, Nichterlein, and
  Sommer}]{DBLP:conf/stacs/KoanaKNS22}
Koana, T.; Komusiewicz, C.; Nichterlein, A.; and Sommer, F. 2022.
\newblock Covering Many (Or Few) Edges with k Vertices in Sparse Graphs.
\newblock In \emph{{STACS}}, volume 219 of \emph{LIPIcs}, 42:1--42:18. Schloss
  Dagstuhl - Leibniz-Zentrum f{\"{u}}r Informatik.

\bibitem[{Koana, Komusiewicz, and
  Sommer(2022)}]{DBLP:journals/siamdm/KoanaKS22}
Koana, T.; Komusiewicz, C.; and Sommer, F. 2022.
\newblock Exploiting {\textdollar}c{\textdollar}-Closure in Kernelization
  Algorithms for Graph Problems.
\newblock \emph{{SIAM} J. Discret. Math.}, 36(4): 2798--2821.

\bibitem[{Koana, Komusiewicz, and
  Sommer(2023{\natexlab{a}})}]{DBLP:journals/algorithmica/KoanaKS23a}
Koana, T.; Komusiewicz, C.; and Sommer, F. 2023{\natexlab{a}}.
\newblock Computing Dense and Sparse Subgraphs of Weakly Closed Graphs.
\newblock \emph{Algorithmica}, 85(7): 2156--2187.

\bibitem[{Koana, Komusiewicz, and
  Sommer(2023{\natexlab{b}})}]{DBLP:journals/algorithmica/KoanaKS23}
Koana, T.; Komusiewicz, C.; and Sommer, F. 2023{\natexlab{b}}.
\newblock Essentially Tight Kernels for (Weakly) Closed Graphs.
\newblock \emph{Algorithmica}, 85(6): 1706--1735.

\bibitem[{Koana and Nichterlein(2021)}]{DBLP:journals/dam/KoanaN21}
Koana, T.; and Nichterlein, A. 2021.
\newblock Detecting and enumerating small induced subgraphs in c-closed graphs.
\newblock \emph{Discret. Appl. Math.}, 302: 198--207.

\bibitem[{Lokshtanov and Surianarayanan(2021)}]{DBLP:conf/fsttcs/LokshtanovS21}
Lokshtanov, D.; and Surianarayanan, V. 2021.
\newblock Dominating Set in Weakly Closed Graphs is Fixed Parameter Tractable.
\newblock In \emph{{FSTTCS}}, volume 213 of \emph{LIPIcs}, 29:1--29:17. Schloss
  Dagstuhl - Leibniz-Zentrum f{\"{u}}r Informatik.

\bibitem[{Mertzios et~al.(2023)Mertzios, Molter, Niedermeier, Zamaraev, and
  Zschoche}]{DBLP:journals/jcss/MertziosMNZZ23}
Mertzios, G.~B.; Molter, H.; Niedermeier, R.; Zamaraev, V.; and Zschoche, P.
  2023.
\newblock Computing maximum matchings in temporal graphs.
\newblock \emph{J. Comput. Syst. Sci.}, 137: 1--19.

\bibitem[{Moon and Moser(1965)}]{moon1965cliques}
Moon, J.~W.; and Moser, L. 1965.
\newblock On cliques in graphs.
\newblock \emph{Israel journal of Mathematics}, 3(1): 23--28.

\bibitem[{Ozella et~al.(2021)Ozella, Paolotti, Lichand, Rodr{\'\i}guez, Haenni,
  Phuka, Leal-Neto, and Cattuto}]{ozella2021using}
Ozella, L.; Paolotti, D.; Lichand, G.; Rodr{\'\i}guez, J.~P.; Haenni, S.;
  Phuka, J.; Leal-Neto, O.~B.; and Cattuto, C. 2021.
\newblock Using wearable proximity sensors to characterize social contact
  patterns in a village of rural Malawi.
\newblock \emph{EPJ Data Science}, 10(1): 46.

\bibitem[{Takaguchi(2015)}]{takaguchi2015analyzing}
Takaguchi, T. 2015.
\newblock Analyzing dynamical social interactions as temporal networks.
\newblock \emph{IFAC-PapersOnLine}, 48(18): 169--174.

\bibitem[{Viard, Latapy, and Magnien(2016)}]{DBLP:journals/tcs/ViardLM16}
Viard, T.; Latapy, M.; and Magnien, C. 2016.
\newblock Computing maximal cliques in link streams.
\newblock \emph{Theor. Comput. Sci.}, 609: 245--252.

\end{thebibliography}

\end{document}